\documentclass[runningheads]{llncs}

\usepackage{amsmath,amssymb,stmaryrd}
\newcommand{\cf}{\mathcal{F}}
\newcommand{\nn}{\mathbb{N}}
\newcommand{\zz}{\mathbb{Z}}

\DeclareMathOperator{\dom}{dom}
\DeclareMathOperator{\width}{width}
\DeclareMathOperator{\height}{height}

\spnewtheorem{thm}[theorem]{Theorem}{\bfseries}{\itshape}
\spnewtheorem{lem}[theorem]{Lemma}{\bfseries}{\itshape}
\spnewtheorem{pro}[theorem]{Proposition}{\bfseries}{\itshape}

\spnewtheorem{dfn}[theorem]{Definition}{\bfseries}{\rm}
\spnewtheorem*{rem}{Remark}{\bfseries}{\rm}

\usepackage{tikz}
\usetikzlibrary{patterns}

\begin{document}

\title{On the entropy of coverable subshifts}

\author{Guilhem Gamard%
  \thanks{This author was funded by Russian Academic Excellence
    Project `5-100'. He has since moved to LIP, ENS Lyon, 46 all\'ee
    d'Italie, 69364, France.}}

\institute{Higher School of Economics. \\
  Office 510, Kochnovsky Proezd 3, Moscow 125319, Russia. \\ %
  Email: \email{guilhem.gamard@normale.fr}}

\authorrunning{G. Gamard}

\maketitle

\begin{abstract}
  A coloration $w$ of $\mathbb{Z}^2$ is said to be \emph{coverable} if
  there exists a rectangular block $q$ such that $w$ is covered with
  occurrences of $q$, possibly overlapping. In this case, $q$ is a
  \emph{cover} of $w$. A subshift is said to have the cover $q$ if
  each of its points has the cover $q$. In a previous article, we
  characterized the covers that force subshifts to be finite (in
  particular, all configurations are periodic). We also noticed that
  some covers force subshifts to have zero topological entropy while
  not forcing them to be finite.  In the current paper we work towards
  characterizing precisely covers which force a subshift to have zero
  entropy, but not necessarily periodicity.  We give a necessary
  condition and a sufficient condition which are close, but not quite
  identical.

\end{abstract}

\section{Introduction}

A \emph{subshift} is a language of infinite words defined by forbidden
factors; for instance, the set of infinite words over $\{a,b\}$ that
do not contain the factors $bb$ nor $aaa$ is a subshift. This specific
kind of language was initially introduced in the context of dynamical
systems~\cite{MorseHedlund1940}; indeed, a subshift equipped
with the ``shift'' action (translate each letter one step to the left)
is a topological dynamical system. On one hand, subshifts viewed as
dynamical systems can be studied with the tools of combinatorics on
words; on the other hand, several systems of interest are conjugate to
subshifts (so they share the same topological invariants). Thus
subshifts make a useful connection between these two fields.

The definition of a subshift is very easy to generalize to higher
dimensions, e.g. to $\zz^2$-words. The \emph{two-dimensional case}
gives rise to a rich theory, that is connected with tilings and
computability. Two-dimensional subshifts may also be used to model
dynamical systems with two commuting actions.

If a subshift is defined by finitely many forbidden factors, then it
has a finite description. Such languages are called \emph{subshifts of
  finite type}, or SFT for short. The restriction to finite type is
worth considering: two-dimensional SFTs are rich enough to encode
objects such as, for instance, Wang tiles, Turing
machines~\cite{Robinson1971}, or physical models such as the square
ice~\cite{Lieb1967}. Most questions that we could ask on SFTs (e.g.,
emptiness, equality) are decidable in dimension~$1$, but undecidable
in dimension~$2$ and higher.

In this article, we focus on two-dimensional SFTs.

A classical problem is to compute the \emph{topological entropy} of
two-dimensional SFTs. From a dynamics point of view, entropy is the
average number of ``information bits'' encoded in each point of the
system. From the combinatorial perspective, on the other hand, entropy
is connected with factor complexity, i.e., the number of factors of
length $n$ (or squares of size $n \times n$ in 2D) occurring in the
subshift in function of $n$.  Finally, topological entropy is
connected with the notion of \emph{residual entropy} in physics.

It is not possible to compute the entropy of an arbitrary SFT $X$,
because we need some information about $X$. This paper considers the
class of \emph{coverable} subshifts. A $\zz^2$-word $w$ \emph{has the
  cover} $q$ if $q$ is a rectangular block and $w$ is covered with
occurrences of $q$, possibly overlapping. A subshift \emph{has the
  cover} $q$ if each of its elements has the cover $q$. (Note that the
term \emph{quasiperiodic} is sometimes used for \emph{coverable} in
the context of one-dimensional words.)

Our motivations to study coverability are threefold.
1. Overlaps in general are an important part of combinatorics on
one-dimensional words. This article is motivated by the larger project
to build a combinatorics on two-dimensional words. Other work
considering two-dimensional overlaps is also
conducted~\cite{AnselmoGiammarresiMadonia2017}, although not in the
context of subshifts.
2. The definition of coverability may be later relaxed; for instance,
we might allow several covers, where each point of a $\zz^2$-word
would have to be covered by one or the other of those covers. By
relaxing the definition more and more, we might understand the entropy
of larger and larger families of SFTs. However, we have to start this
project with the most constrained definition: one cover of rectangular
shape.
3. The famous Penrose tilings~\cite{Penrose1978} can be described in
terms of a single tile that overlaps itself~\cite{Gummelt1996}.  In
one dimension, the family of standard Sturmian words (which are
right-infinite words) can be characterized in terms of
covers~\cite{GamardRichomme2016}. These examples show that coverable
phenomena are found in otherwise natural examples of words, so it
makes sense to port this definition to SFTs.

Coverability was initially defined on finite words, in the context of
text algorithms~\cite{ApostolicoEhrenfeucht1993}; it was subsequently
generalized to infinite words and
1D-subshifts~\cite{Marcus2004,MarcusMonteil2006}. In parallel, its
connections with morphisms on words and Sturmian words were
throughoutly studied~\cite[and its
references]{GlenLeveRichomme2008}. All this work in one dimension will
provide us with ideas and techniques to study the two-dimensional
case, but the generalization to higher dimensions is far from trivial
and new ideas are also required.

In two dimensions, an efficient algorithm to find the covers of any
finite, square-word is
known~\cite{CrochemoreIliopoulosKorda1998}. Besides, a few properties
of coverable 2D SFTs were proven in previous articles, notably about
entropy, minimality and uniform
frequencies~\cite{GamardRichomme2015,GamardRichomme2017}.

This paper is structured as follows. In Section~\ref{sec:prel}, after
a quick review of basic definitions and notation, we prove that the
language of $q$-coverable configurations is an SFT for all $q$. Then
we try to compute its entropy in function of $q$. Since exact values
are difficult to obtain, we engage in a simpler task: characterize
which $q$'s yield zero entropy and which yield strictly positive
entropy. In Section~\ref{sec:nec}, we give a necessary condition for
strictly positive entropy. Then, in Section~\ref{sec:inter}, we define
\emph{interchangeable pairs}, and show how this notion is useful to
compute enropies of coverable subshifts. Finally, in
Section~\ref{sec:suf}, we give a sufficient condition for positive
entropy, which is close to (but not quite) the negation of the
necessary condition; we also give a lower bound on the entropy for
subshifts that satisfy the sufficient condition. Our conclusion is
Section~\ref{sec:conclu}: we give a few open problems and state our
acknowledgements.

\section{Preliminaries}
\label{sec:prel}

We start by reviewing definitions and notation. A \emph{configuration}
is a coloring of $\zz^2$ whose colors are taken from some finite
alphabet $\Sigma$. A \emph{domain} is a finite subset of $\zz^2$ and a
\emph{fragment} is a coloring of a domain. When we consider a fragment
up to translation, i.e., we are not interested in its position in the
plane, we call it a \emph{pattern}.  A \emph{block} is a pattern $p$
whose domain is a rectangle, i.e., there exists natural (nonnegative)
integers $m,n$ such that
$\dom(p)=\{0, \ldots, m-1\} \times \{0, \ldots, n-1\}$. The number $m$
is the \emph{width} and the number $n$ the \emph{height} of the
rectangle. The \emph{position} of a block is the position of its
bottom, left-hand corner. We note $\Sigma^{m \times n}$ the set of all
blocks of size $m \times n$ over alphabet $\Sigma$. Here are
the intuitive correspondences with the unidimensional case: \\
\centerline{
\begin{tabular}{l l l}
  configuration & $\iff$ & infinite word \\
  pattern, block & $\iff$ & finite word \\
  fragment & $\iff$ & occurrence of a finite word in an infinite word
\end{tabular}} \\
If $D$ is a set (in particular a domain), then $|D|$ denotes the
cardinality of $D$. If $u$ is either a pattern or a fragment, then
$|u|$ denotes the cardinality of $\dom(u)$.

Let $u$ denote a pattern, $w$ a configuration or a pattern, and $D$ a
domain. The notation $w(D)$ refers to the restriction of $w$ to domain
$D$. If $u=w(D)$, (so in particular $\dom(u)=D$ up to translation),
then we say that $u$ \emph{occurs in} $w$.

Let $f$ denote a fragment. Elements of $\zz^2$ are often called
\emph{positions}; if $(i,j)$ belongs to $\dom(f)$, then we say that
$f$ \emph{covers} the position $(i,j)$. Two fragments said to be
\emph{neighbouring} if the union of their domains is simply connected
(counting only vertical and horizontal neighbours) and if they agree
on the intersections of their domains (which might be empty). If
moreover the intersection of their domains is not empty, then we say
that they \emph{overlap}.

\begin{dfn}
  Let $q$ denote a block. A fragment, pattern, or a configuration $w$
  is said to be $q$-\emph{coverable} if each position of its domain is
  covered by a copy of $q$. Formally, there exist domains
  $D_1, \ldots, D_n$ (possibly $n=\infty$) such that
  $\dom(w)=\bigcup_{i=1}^n D_i$ and $w(D_i)$ is equal to $q$ up to
  translation for all $i$.
\end{dfn}

\begin{dfn}
  A set of configurations $X$ is a subshift if and only if there
  exists a set of patterns $\cf$ such that $X$ is the set of
  configurations in which no element of $\cf$ occur, i.e.,
  $X = \{x \in \Sigma^{\zz^2} \,|\, \forall D \subseteq \zz^2, x(D)
  \not\in \cf\}$.

  Note that two different sets of forbidden patterns might yield the
  same subshift. If $\cf$ can be made finite, then $X$ is said to be a
  \emph{subshift of finite type} (or SFT for short).
\end{dfn}

Each subshift is stable by translation: if $X$ is a subshift,
$x \in X$, $y \in \Sigma^{\zz^2}$ and there is a vector $\vec{v}$ such that
$y(\vec{u}) = x(\vec{u}+\vec{v})$ for all $\vec{u}$, then $y \in X$.

\begin{pro} \label{pro:qpsft}
  Given a block $q$, the set of all $q$-coverable configurations
  is a subshift of finite type, that we note $X_q$.
\end{pro}

\begin{rem} \label{rem:easyq} %
  Let $w$ denote a fragment or a configuration and $(x,y)$ the
  position of a copy of $q$ in $w$. Then that copy of $q$ covers the
  position $(0,0)$ if and only if $-|w|+1 \leq x \leq 0$ and
  $-|h|+1 \leq y \leq 0$.
\end{rem}

\begin{proof}[of Proposition~\ref{pro:qpsft}]
  Let $(w,h)$ denote the dimensions of $\dom(q)$.  Define the domain
  $D= \{-|w|+1, \dots, |w|-1\} \times \{-|h|+1, \dots, |h|-1\}$ and
  $\cf$ the set of patterns with domain $D$ that do not contain any
  occurrence of $q$. We show that $q$-coverable configurations are
  exactly the configurations that avoid the patterns in $\cf$.

  In an arbitrary configuration $w$ covered by $q$, each block of size
  $(2|w|-1, 2|h|-1)$ contains at least one occurrence of $q$. Suppose
  not; we can assume without loss of generality that our faulty block
  is at position $(-|w|+1,-|h|+1)$. Then, by the remark above, the
  position $(0,0)$ is not covered by $q$ in $w$: a contradiction.

  Conversely, if $w$ is a configuration \emph{not} covered by $q$,
  then some pattern from $\cf$ occurs in it. Indeed, there exists a
  position in $w$ which is not covered by $q$; assume without loss of
  generality that this position is $(0,0)$. Then by the remark above,
  $w(D)$ does not contain any occurrence of $q$, so it belongs to
  $\cf$.

  As a conclusion, $X_q$ is the set of configurations that do not
  contain any pattern in $\cf$. Since all forbidden blocks have domain
  $D$, which is a finite set, the set $\cf$ itself is finite, so $X_q$
  is a subshift of finite type.
\end{proof}

\begin{dfn}
  If $r$ is a block and $m,n$ natural integers, we call
  $r^{m \times n}$ the pattern made of $m$ copies of $r$ concatenated
  horizontally, repeated $n$ times vertically. If $q$ can be written
  $r^{n \times m}$ for some strictly positive integers $m,n$, then we
  say that $r$ is a \emph{root} of $q$. Since $q = q^{1 \times 1}$,
  the block $q$ is always a root of $q$. If $q$ has no root besides
  itself, we say that it is \emph{primitive}.
\end{dfn}

When $q$ is nonprimitive, it has an unique primitive root which is
root of all other roots~\cite[Lemma~4]{GamardRichomme2015}.

\begin{dfn}
  If $w$ and $u$ are two different blocks such that $u$ occurs in two
  opposite corners of $w$, then $u$ is called a \emph{border} of
  $w$.
\end{dfn}

For instance, $a$, $b$ and $b^{2 \times 2}$ are borders of
$\;\begin{smallmatrix}
  a & b & b \\
  b & b & b \\
  b & b & a
\end{smallmatrix}$. Observe that a block is never a border of itself,
but the empty block is a border of all blocks.

\begin{thm}[{See~\cite[Theorem~5]{GamardRichomme2015}}]
  Let $q$ denote a block; then $X_q$ is infinite if and only if the
  primitive root of $q$ has a nonempty border.
\end{thm}

Generally speaking we are interested in getting more results which
describe $X_q$, given properties on $q$. In this paper, we focus on
the \emph{topological entropy} which is, intuitively, the average
number of bits necessary to encode one cell of one configuration in
$X$.

\begin{dfn}
  Let $X$ denote a subshift and $L(X)$ the set of all blocks occuring
  in a configuration of $X$, i.e.,
  $L(X) = \{p \;|\; \exists x \in X,\; p \mbox{ occurs in } x\}$.
  Define $L_{m,n}(X)$ to be the set of blocks of size $m \times n$
  occurring in $X$, that is to say,
  $L_{m,n}(X) = L(X) \cap \Sigma^{m \times n}$. The \emph{topological
    entropy} of $X$ is the number:
  \begin{equation*}
    h(X) = \lim_{n \to \infty} \frac{\log_2(|L_{n,n}(X)|)}{n^2}
  \end{equation*}
\end{dfn}

This limit always exists. (In the one-dimensional case, the sequence
$\log(|L_n(X)|)$ is subadditive, so by the subadditivity lemma
$\log(|L_n(X)|)/n$ coverges. The passage to higher dimensions boils
down to a computation.) Observe that if
$|L_{n,n}(X)| \sim \varepsilon^{n^2}$ for some positive $\varepsilon$,
then $h(X) = \log_2 \varepsilon$. Otherwise, if $|L_{n,n}(X)|$ grows
slower than any $\varepsilon^{n^2}$, then $h(X) = 0$. Thus the maximal
value for the topological entropy is $\log_2 |\Sigma|$, because there
is at most $|\Sigma|^{n^2}$ elements in $L_{n,n}(X)$.

\section{A necessary condition for strictly positive entropy}
\label{sec:nec}

\begin{dfn}
  \label{dfn:starcond}
  If a block has two borders in consecutive corners that are neighbouring
  (see Figure~\ref{fig:starcond}), then we say that it satisfies the condition~$(*)$.
\end{dfn}

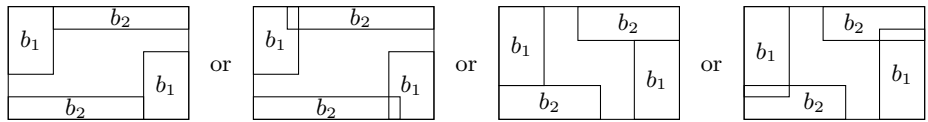
\begin{figure}[htbp]
  \centering
  \begin{tikzpicture}[scale=0.3,baseline={([yshift=-.8ex]current bounding box.center)}]
    \draw (0,0) rectangle (8,5);
    \draw (0,5) rectangle node{$b_1$} ++(2,-3);
    \draw (8,0) rectangle node{$b_1$} ++(-2,3);
    \draw (0,0) rectangle node{$b_2$} ++(6,1);
    \draw (8,5) rectangle node{$b_2$} ++(-6,-1);
  \end{tikzpicture}
  \hfill
  or
  \hfill
  \begin{tikzpicture}[scale=0.3,baseline={([yshift=-.8ex]current bounding box.center)}]
    \draw (0,0) rectangle (8,5);
    \draw (0,5) rectangle node{$b_1$} ++(2,-3);
    \draw (8,0) rectangle node{$b_1$} ++(-2,3);
    \draw (0,0) rectangle node{$b_2$} ++(6.5,1);
    \draw (8,5) rectangle node{$b_2$} ++(-6.5,-1);
  \end{tikzpicture}
  \hfill
  or
  \hfill
  \begin{tikzpicture}[scale=0.3,baseline={([yshift=-.8ex]current bounding box.center)}]
    \draw (0,0) rectangle (8,5);
    \draw (0,5) rectangle node{$b_1$} ++(2,-3.5);
    \draw (8,0) rectangle node{$b_1$} ++(-2,3.5);
    \draw (0,0) rectangle node{$b_2$} ++(4.5,1.5);
    \draw (8,5) rectangle node{$b_2$} ++(-4.5,-1.5);
  \end{tikzpicture}
  \hfill
  or
  \hfill
  \begin{tikzpicture}[scale=0.3,baseline={([yshift=-.8ex]current bounding box.center)}]
    \draw (0,0) rectangle (8,5);
    \draw (0,5) rectangle node{$b_1$} ++(2,-4);
    \draw (8,0) rectangle node{$b_1$} ++(-2,4);
    \draw (0,0) rectangle node[right]{$b_2$} ++(4.5,1.5);
    \draw (8,5) rectangle node[left]{$b_2$} ++(-4.5,-1.5);
  \end{tikzpicture}
  \caption{Illustration of the condition~$(*)$.}
  \label{fig:starcond}
\end{figure}

Let $q$ denote an arbitrary block. If $q$ has a full-width or
full-height border, i.e. a border having the same width (resp. the
same height) as $q$ itself, then $q$ satisfies the
condition~$(*)$. Indeed, this full-width or full-height border occurs
in two consecutive corners, and it obviously overlaps with itself.

\begin{thm}
  \label{thm:main}
  Let $q$ denote a primitive block. If $X_q$ has positive entropy, then
  $q$ satisfies the condition~$(*)$.
\end{thm}

The contraposition is also interesting: if $q$ does not satisfy the
condition~$(*)$, then $X_q$ has zero entropy.

If $q$ is nonprimitive, then it has a full-width or a full-height
border and thus it satisfies condition~$(*)$. Besides, if $q$ covers
some pattern or configuration $w$, then so does each root of
$q$. Therefore it makes sense to restrict ourselves to primitive
covers. We could lift this restriction by replacing, in
condition~$(*)$ and elsewhere, the term \emph{border} with
\emph{border that is not a power of the primitive root}. This would
require extra precautions in the proofs while not adding any
significant value to our results, so we will keep the supposition that
$q$ is primitive.

The remainder of this section is devoted to the proof of
Theorem~\ref{thm:main}.

\begin{dfn}
  Let $w$ denote a configuration with a cover $q$. We note by
  $\rho(i,j)$ the topmost among the rightmost occurrences of $q$
  covering position $(i,j)$.
\end{dfn}

This means that we first select the rightmost copies of $q$ covering
position $(i,j)$, and then among them (if there are several) we select
the topmost one. Observe that if $q$ does not satisfy the
condition~$(*)$, then the ``rightmost occurrence of $q$ covering
position $(i,j)$'' is unique: we don't need to select the topmost
one. Indeed, if we had two rightmost occurrences covering $(i,j)$,
they would share the same $x$-coordinate, and either they would be
equal, or $q$ would have a full-width border: as explained above, this
implies the condition~$(*)$.

The vector $\rho$ implicitly depends on $q$ and $w$, but there will be
no ambiguity in what follows.

\begin{lem}
  \label{lem:combi}
  Let $q$ denote a primitive block of size $m \times n$ that does not
  satisfy condition~$(*)$, and $w$ a configuration covered by
  $q$. Then $\rho(0,0)$ is either equal to $(0,0)$, or to
  $\rho(-1,0)$, or to $\rho(0,-1)$. In particular, $\rho(0,0)$ is
  uniquely determined by $\rho(-1,0)$ and $\rho(0,-1)$.
\end{lem}
\begin{proof}
  Call $(x,y)$ and $(x',y')$ the positions of $\rho(-1,0)$ and
  $\rho(0,-1)$, respectively, and $(i,j)$ the position of $\rho(0,0)$.
  We have $i \leq 0$ and $j \leq 0$, otherwise $\rho(0,0)$ would not
  cover the position $(0,0)$. There are several cases to consider.

    \begin{figure}[p]
    \centering
    \begin{minipage}[b]{1in}
      \centering
      \begin{tikzpicture}[scale=0.15]
        \draw (0,0) rectangle node[below]{$q$} ++(5,8);
        \draw (-3.5,4.5) rectangle node[above]{$q$} ++(5,8);
        \draw (1.5,8) ++(1.0em,1.0em) node{$\circ$};
      \end{tikzpicture}
      \caption{}\label{subfig:combi1a}	
    \end{minipage}
    \hfill
    \begin{minipage}[b]{1in}
      \centering
      \begin{tikzpicture}[scale=0.15]
        \draw (0,0) rectangle node[below]{$q$} ++(5,8);
        \draw (-3.5,4.5) rectangle node[above]{$q$} ++(5,8);
        \draw (1.5,7) rectangle node[right]{$q$} ++(5,8);
        \draw[pattern=north west lines] (0,4.5) rectangle ++(1.5,3.5);
        \draw[pattern=north east lines] (1.5,7) rectangle ++(3.5,1.0);
      \end{tikzpicture}
      \caption{}\label{subfig:combi1b}
    \end{minipage}
    \hfill
    \begin{minipage}[b]{1in}
      \centering
      \begin{tikzpicture}[scale=0.15]
        \draw (0,0) rectangle node[below]{$q$} ++(5,8);
        \draw (-3.5,4.5) rectangle node[above]{$q$} ++(5,8);
        \draw (-0.5,8) rectangle node[right]{$q$} ++(5,8);
        \draw[pattern=north west lines] (0,4.5) rectangle ++(1.5,3.5);
        \draw[pattern=north east lines] (-0.5,8) rectangle ++(2,4.5);
      \end{tikzpicture}
      \caption{}\label{subfig:combi1c}
    \end{minipage}
    \hfill
    \begin{minipage}[b]{1in}
      \centering
      \begin{tikzpicture}[scale=0.15]
        \draw (0,0) rectangle node[below]{$q$} ++(5,8);
        \draw (-3,6) rectangle node[above]{$q$} ++(5,8);
        \draw (-5,-2) rectangle node[below]{$q$} ++(5,8);
        \draw (-1.5,3) rectangle ++(5,8);
        \draw[pattern=north east lines] (-1.5, 3) rectangle ++(1.5,3);
        \draw[pattern=north west lines] (0, 3) rectangle ++(3.5,5);
        \draw (2,8) ++(0.5,0.5) node{$\circ$};
        \draw (-1,5) ++(0.5,0.5) node{$\bullet$};
      \end{tikzpicture}
      \caption{}\label{subfig:combi1d}
    \end{minipage}
    \begin{center}
      \emph{Case 3.} The origin of the occurrence covering $\circ$ is
      forced to be $\circ$, otherwise the condition~$(*)$ is satisfied
      (see hatched lines).
    \end{center}

    \begin{minipage}[b]{1in}
      \centering
      \begin{tikzpicture}[scale=0.15]
        \draw (0,0) rectangle node{$q$} ++(5,8);
        \draw (-5,3) rectangle node{$q$} ++(5,8);
        \draw (0,8) ++(0.5,0.5) node{$\circ$};
      \end{tikzpicture}
      \caption{}\label{subfig:combi2a}
    \end{minipage}
    \hfill
    \begin{minipage}[b]{1in}
      \centering
      \begin{tikzpicture}[scale=0.15]
        \draw (0,0) rectangle node{$q$} ++(5,8);
        \draw (-5,3) rectangle node{$q$} ++(5,8);
        \draw (-1.5,6) rectangle node[right]{$q$} ++(5,8);
        \draw (0,8) ++(0.5,0.5) node{$\circ$};
        \draw[pattern=north east lines] (-1.5,6) rectangle ++(1.5,5);
        \draw[pattern=north west lines] (0,6) rectangle ++(3.5,2);
      \end{tikzpicture}
      \caption{}\label{subfig:combi2b}
    \end{minipage}
    \hfill
    \begin{minipage}[b]{1in}
      \centering
      \begin{tikzpicture}[scale=0.15]
        \draw (0,0) rectangle node[below]{$q$} ++(5,8);
        \draw (-5,6) rectangle node{$q$} ++(5,8);
        \draw (-1.5,4) rectangle ++(5,8);
        \draw (-5,-2) rectangle node{$q$} ++(5,8);
        \draw (-1,5) ++(0.5,0.5) node{$\bullet$};
        \draw[pattern=north east lines] (-1.5,4) rectangle ++(1.5,2);
        \draw[pattern=north west lines] (-1.5,6) rectangle ++(1.5,6);
        \draw[pattern=north west lines] (0,4) rectangle ++(3.5,4);
        \draw (0,8) ++(0.5,0.5) node{$\circ$};
      \end{tikzpicture}
      \caption{}\label{subfig:combi2c}      
    \end{minipage}
    \hfill
    \begin{minipage}[b]{1in}
      \centering
      \begin{tikzpicture}[scale=0.15]
        \draw (0,0) rectangle node[below]{$q$} ++(5,8);
        \draw (-5,3) rectangle node{$q$} ++(5,8);
        \draw (0,6) rectangle node[above]{$q$} ++(5,8);
        \draw[pattern=north west lines] (0,6) rectangle ++(5,2);
        \draw (0,8) ++(0.5,0.5) node{$\circ$};
        \end{tikzpicture}
      \caption{}\label{subfig:combi2d}
    \end{minipage}

    \begin{center}
      \emph{Case 4.} The origin of the occurrence convering $\circ$ is
      forced.
    \end{center}

    \begin{minipage}[b]{1in}
      \centering
      \begin{tikzpicture}[scale=0.15]
        \draw (1,0) rectangle node{$q$} ++(5,8);
        \draw (-2,8) rectangle node{$q$} ++(5,8);
        \draw (3,8) ++(0.5,0.5) node{$\circ$};
      \end{tikzpicture}
      \caption{}\label{subfig:combi3a}
    \end{minipage}
    \hfill
    \begin{minipage}[b]{1in}
      \centering
      \begin{tikzpicture}[scale=0.15]
        \draw (1,0) rectangle node{$q$} ++(5,8);
        \draw (-2,8) rectangle node{$q$} ++(5,8);
        \draw (2,5) rectangle node[above]{$q$} ++(5,8);
        \draw[pattern=north west lines] (2,5) rectangle ++(4,3);
        \draw[pattern=north east lines] (2,8) rectangle ++(1,5);
        \draw (3,8) ++(0.5,0.5) node{$\circ$};
      \end{tikzpicture}
      \caption{}\label{subfig:combi3b}
    \end{minipage}
    \hfill
    \begin{minipage}[b]{1in}
      \centering
      \begin{tikzpicture}[scale=0.15]
        \draw (1,0) rectangle node{$q$} ++(5,8);
        \draw (-4,0) rectangle node{$q$} ++(5,8);
        \draw (-2,8) rectangle ++(5,8);
        \draw (0,7) rectangle ++(5,8);
        \draw[pattern=north west lines] (0,7) rectangle ++(1,1);
        \draw[pattern=north east lines] (1,7) rectangle ++(4,1);
        \draw[pattern=north east lines] (0,8) rectangle ++(3,7);
        \draw (3,8) ++(0.5,0.5) node{$\circ$};
      \end{tikzpicture}
      \caption{}\label{subfig:combi3c}
    \end{minipage}
    \hfill
    \begin{minipage}[b]{1in}
      \centering
      \begin{tikzpicture}[scale=0.15]
        \draw (1,0) rectangle node{$q$} ++(5,8);
        \draw (-2,8) rectangle node[left]{$q$} ++(5,8);
        \draw (1.5,8) rectangle node[right]{$q$} ++(5,8);
        \draw[pattern=north west lines] (1.5,8) rectangle ++(1.5,8);
        \draw (3,8) ++(0.5,0.5) node{$\circ$};
      \end{tikzpicture}
      \caption{}\label{subfig:combi3d}
    \end{minipage}
    \begin{center}
      \emph{Case 5.} The origin of the occurrence covering $\circ$ is
      forced.
    \end{center}
    \label{fig:combi3}

    \begin{minipage}[b]{1in}
      \centering
      \begin{tikzpicture}[scale=0.15]
        \draw (-5,0) rectangle node{$q$} ++(5,8);
        \draw (0,-8) rectangle node{$q$} ++(5,8);
        \draw (0.5,0.5) node{$\circ$};
      \end{tikzpicture}
      \caption{}\label{subfig:combi4a}
    \end{minipage}
    \hfill
    \begin{minipage}[b]{1in}
      \centering
      \begin{tikzpicture}[scale=0.15]
        \draw (-5,0) rectangle ++(5,8);
        \draw (0,-8) rectangle ++(5,8);
        \draw (-3,-3) rectangle ++(5,8);
        \draw[pattern=north east lines] (-3,0) rectangle ++(3,5);
        \draw[pattern=north east lines] (0,-3) rectangle ++(2,3);
        \draw (0,0) ++(0.5,0.5) node{$\circ$};
        \draw (2,0) ++(0.5,0.5) node{$\star$};
        \draw (0,5) ++(0.5,0.5) node{$\bullet$};
      \end{tikzpicture}
      \caption{}\label{subfig:combi4b}
    \end{minipage}
    \hfill
    \begin{minipage}[b]{1in}
      \centering
      \begin{tikzpicture}[scale=0.15]
        \draw (-5,0) rectangle ++(5,8);
        \draw (0,-8) rectangle ++(5,8);
        \draw (-3,-3) rectangle ++(5,8);
        \draw[pattern=north east lines] (-3,0) rectangle ++(3,5);
        \draw[pattern=north east lines] (0,-3) rectangle ++(2,3);
        \draw (2,-4) rectangle ++(5,8);
        \draw (-2,5) rectangle ++(5,8);
        \draw[pattern=north west lines] (2,-4) rectangle ++(3,4);
        \draw[pattern=north west lines] (-2,5) rectangle ++(2,3);
        \draw (0,0) ++(0.5,0.5) node{$\circ$};
        \draw (2,0) ++(0.5,0.5) node{$\star$};
        \draw (0,5) ++(0.5,0.5) node{$\bullet$};
      \end{tikzpicture}
      \caption{}\label{subfig:combi4c}
    \end{minipage}
    \hfill
    \begin{minipage}[b]{1in}
      \centering
      \begin{tikzpicture}[scale=0.15]
        \draw (-5,0) rectangle ++(5,8);
        \draw (0,-8) rectangle ++(5,8);
        \draw (-3,-3) rectangle ++(5,8);
        \draw[pattern=north east lines] (0,0) rectangle ++(2,5);
        \draw (2,0) rectangle ++(5,8);
        \draw (0,5) rectangle ++(5,8);
        \draw[pattern=north west lines] (2,5) rectangle ++(3,3);
        \draw (2,0) ++(0.5,0.5) node{$\star$};
        \draw (0,5) ++(0.5,0.5) node{$\bullet$};
      \end{tikzpicture}
      \caption{}\label{subfig:combi4d}
    \end{minipage}
    \begin{center}
      \emph{Case 6.} The origin of the occurrence covering $\circ$ is
      forced.
    \end{center}
    \label{fig:combi4}
  \end{figure}

  \smallskip

  \emph{Case 1.} If $(x,y)=(x',y')$, then necessarily we have
  $(x,y)=(x',y')=(i,j)$ by definition of $\rho$ and because $\dom(q)$
  is convex (it is a rectangle).

  \smallskip
  
  \emph{Case 2.} If $\rho(-1,0)$ (respectively $\rho(0,-1)$) contains
  $(0,0)$, then $\rho(0,0)=\rho(-1,0)$ (respectively
  $\rho(0,0)=\rho(0,-1)$) so $\rho(0,0)$ is uniquely determined.
  
  \smallskip
  
  Now we have to consider the cases where the occurrences
  $\rho(0,-1)$, $\rho(-1,0)$ and $\rho(0,0)$ are all different. We
  will prove that, in this situation, the position of $\rho(0,0)$ is
  always $(0,0)$. There are two disjoint cases to consider: either
  $\rho(0,-1)$ and $\rho(-1,0)$ overlap, or they don't.
  
  \smallskip

  \emph{Case 3.}  Suppose that $(x,y)$, $(x',y')$, and $(i,j)$ are all
  different and that the occurrences of $q$ at positions $(x,y)$ and
  $(x',y')$ overlap. We are in the situation of
  Figure~\ref{subfig:combi1a}, where coordinate $(0,0)$ is marked by
  the symbol $\circ$. If $i=0$ and $j < 0$, then
  Figure~\ref{subfig:combi1b} shows how $q$ would satisfy the
  condition~$(*)$: the hatched areas are the two consecutive borders
  which are neighbouring. If $i < 0$ and $j=0$, then
  Figure~\ref{subfig:combi1c} shows how $q$ would satisfy the
  condition~$(*)$. Finally, if both $j < 0$ and $i < 0$, then we are
  in the situation of Figure~\ref{subfig:combi1d}.  Consider
  $(x'', y'') = \rho(x'-1, y-1)$; the coordinates $(x'-1, y-1)$ are
  shown by the symbol $\bullet$ on Figure~\ref{subfig:combi1d}. Then
  the figure shows how the condition~$(*)$ would be satisfied
  again. In this figure, $x''=x'-m$ and $y''=y-n$, but the argument
  works even if $x''$ or $y''$ or both were larger than that.
  
  \smallskip

  We still have to analyse the case where $\rho(-1,0)$ and
  $\rho(0,-1)$ do \emph{not} overlap. There are three possibilities
  here: either $x-x'=m$ (as in Figure~\ref{subfig:combi2a}), or
  $y'-y=n$ (as in Figure~\ref{subfig:combi3a}), or both (as in
  Figure~\ref{subfig:combi4a}). As before, we prove that
  $\rho(0,0)=(0,0)$.
  
  \smallskip

  \emph{Case 4.} Suppose that $\rho(-1,0)$, $\rho(0,-1)$, and
  $\rho(0,0)$ are all different, that $\rho(-1,0)$ and $\rho(0,-1)$ do
  not overlap and that $x'-x=m$. We are in the situation of
  Figure~\ref{subfig:combi2a}. If $i<0$ and $y \leq j < 0$ then the
  condition~$(*)$ is satisfied, as shown on
  Figure~\ref{subfig:combi2b}. If $i<0$ and $j<y$, then let
  $(x'',y'')$ denote the position of $\rho(x'-1,y-1)$; the coordinates
  $(x'-1,y-1)$ are marked by the symbol $\bullet$ on
  Figure~\ref{subfig:combi2c}. This figure shows how the
  condition~$(*)$ is satisfied (on the figure, the minimal $x''$ and
  $y''$ are represented, but the argument still works if $x''$ or
  $y''$ or both are larger). If $i=0$ and $j<0$ then $q$ has a
  full-width border, as shown on Figure~\ref{subfig:combi2d}, so the
  condition~$(*)$ is also satisfied. The only remaining case is
  $(i,j)=(0,0)$.
  
  \smallskip
  
  \emph{Case 5.} Suppose that $\rho(-1,0)$, $\rho(0,-1)$, and
  $\rho(0,0)$ are all different, that $\rho(-1,0)$ and $\rho(0,-1)$ do
  not overlap and that $y'-y=n$. We are in the situation of
  Figure~\ref{subfig:combi3a}. The reasonning is similar to the
  previous case, but the roles of $x$ and $y$ are swapped; see
  Figures~\ref{subfig:combi3b}, \ref{subfig:combi3c}
  and~\ref{subfig:combi3d}.
  
  \smallskip

  \emph{Case 6.} $(x,y) = (-m,0)$ and $(x',y') = (0,-n)$; we are in
  the situation of Figure~\ref{subfig:combi4a}. We show that
  $(i,j)=(0,0)$, so suppose towards a contradiction that
  $(i,j) \neq (0,0)$, as on Figure~\ref{subfig:combi4b}. Let
  $(k,\ell) = \rho(i+m,0)$ and $(k',\ell')=\rho(0,j+n)$; the
  coordinates $(i+m,0)$ and $(0,j+n)$ are respectively shown as a
  $\bullet$ and a $\star$ on Figure~\ref{subfig:combi4b}. If either
  $k<0$ or $\ell'<0$, then the condition~$(*)$ would be satisfied, as
  shown by Figure~\ref{subfig:combi4c}. The only remaining possibility
  is $k=\ell'=0$. But then the condition~$(*)$ is satisfied again, as
  shown on Figure~\ref{subfig:combi4d}.

  \smallskip

  \emph{Conclusion.} We showed that, in each case, either
  $\rho(0,0)=(0,0)$, or $\rho(0,0)=\rho(-1,0)$, or
  $\rho(0,0)=\rho(0,-1)$.  Thus, $\rho(0,0)$ is uniquely determined by
  $\rho(-1,0)$ and $\rho(0,-1)$, and the lemma is proved.
\end{proof}

\begin{proof}[of Theorem~\ref{thm:main}]
  We prove the contraposition of the theorem: if $q$ does not satisfy
  the condition~$(*)$, then $X_q$ has zero topological entropy. Recall
  that $q$ is of size $m \times n$. Consider an arbitrary square
  occurring in a configuration of $X_q$, i.e., an element of
  $L_{k,k}(X_q)$ for some $k \in \mathbb{N}$. This square appears in a
  configuration $w$ in $X_q$, and we can assume without loss of
  generality that it has position $(0,0)$. (Indeed, if a configuration
  belongs to $X_q$, so do all the translations of that configuration.)
  We will bound the number of possibilities for such a square.

  Let $I$ denote $\{(-k,k), (-k+1,k-1), \ldots, (k,-k)\}$ (the darkest
  cells in Figure~\ref{fig:IIpIpp}). Suppose that we know $\rho(i)$
  for each $i$ in $I$. By applying Lemma~\ref{lem:combi} several
  times, we can uniquely determine $\rho(i')$ for each $i'$ in
  $I'=\{(-k+1,k), \ldots, (k,-k+1) \}$ (this is $I$ shifted one cell
  to the right, minus the bottommost cell). This information, in turn,
  determines $\rho(i'')$ for each $i''$ in
  $I''=\{(-k+2,k), \ldots, (k,-k+2)\}$, and so on. By iterating this
  process, we deduce the contents of the whole square
  $\{0, \ldots, k-1\} \times \{0, \ldots, k-1\}$, and even a bit more
  (see Figure~\ref{fig:IIpIpp}). The area that we can determine is not
  shaped like a square in general, and there might have several
  $k \times k$ squares in it. We can locate the desired square with
  two coordinates $(x,y)$ satisfying $0 \leq x,y \leq 2k+1$.

  \begin{figure}[htbp]
    \centering
    \begin{tikzpicture}[scale=0.35]
      \draw[gray] (-4.5,-4.5) grid (6.5,6.5);
      \foreach \i in {-3,...,3} {
        \fill[gray] (-\i,\i) rectangle ++(1,1);
      }
      \foreach \i in {-2,...,3} {
        \fill[gray!50!white] (-\i,\i) ++(1,0) rectangle ++(1,1);
      }
      \foreach \i in {-1,...,3} {
        \fill[gray!75!white] (-\i,\i) ++(2,0) rectangle ++(1,1);
      }
      \draw[very thick] (0,0) rectangle ++(4,4);
      \draw (-2.6,4.5) node{$I$};
      \draw (-1.6,4.5) node{$I'$};
      \draw (-0.6,4.5) node{$I''$};
      \draw (0.9,4.5) node{$\dots$};
      \fill[black] (0,0) circle (5pt);
      \draw (0,0) node[below]{\tiny$0$};
      \draw[thick,<->] (4.5,0) -- node[right]{$k$} (4.5,4);
    \end{tikzpicture}
    \caption{From $\rho(I)$ we deduce $\rho(I'), \rho(I''), \ldots$ and in the end we know
      the contents of the square.}
    \label{fig:IIpIpp}
  \end{figure}
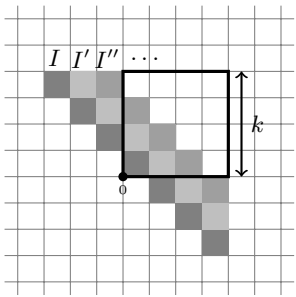

  For each $i$ in $I$, there are $nm$ possibilities for $\rho(i)$, so
  we can compute a bound on the number of $k \times k$ squares in
  $X_q$: we have not more than $u_k = (2k+1)^2 \times (nm)^{2k+1}$
  such squares. The sequence $\log(u_k)/k^2$ converges to $0$ as $k$
  grows to infinity, so the entropy of $X_q$ must be zero. The theorem
  is proved.
\end{proof}

\section{Interchangeable pairs}
\label{sec:inter}

Now our goal is to give a sufficient condition on a block $q$ to force
$X_q$ to have strictly positive entropy. We use a tool called
\emph{interchangeable pairs}, which we define now. In what follows, a
\emph{$q$-patch} is a $q$-coverable pattern.

\begin{dfn}
  \label{dfn:interchg}
  An \emph{interchangeable pair for $q$} is a pair of different
  $q$-patches with the same domain.
\end{dfn}

Let $p_1, p_2$ be an interchangeable pair for $q$ and $w$ a
configuration in $X_q$. Any occurrence of $p_1$ in $w$ can be replaced
with an occurrence of $p_2$; the result would still be a configuration
of $X_q$ (and different from $w$). Hance the name interchangeable
pair.

\begin{pro}
  \label{pro:posent}
  Let $q$ denote a primitive block. If the subshift $X_q$ has strictly
  positive entropy, then there exists an interchangeable pair for $q$.
\end{pro}
\begin{proof}
  By contraposition: suppose that there is no interchangeable pair for
  $q$. Let $n$ denote an integer and $Y$ the set of $q$-patches whose
  domains contain a square of size $n \times n$ and minimal for this
  property. In other terms, each $y$ in $Y$ contains a square of size
  $n \times n$, but if we remove one occurrence of $q$ from $y$, then
  it is not the case anymore. Since there is no interchangable pair
  for $q$, each element of $Y$ is uniquely determined by the shape of
  its domain, thus $Y$ contains not more than $|q|^{4n}$
  elements. Ineed, the shape can be uniquely determined by attaching,
  for each cell in the frontier of the $n \times n$-square, the
  position of an occurrence of $q$ relative to the cell. Any block $c$
  in $L_{n,n}(X_q)$ is uniquely determined by an element of $Y$ and a
  pair of coordinates $(x,y)$ satisfying $0 \leq x,y \leq n + 2|q|-1$,
  so there are not more than
  $v_n = (n+2|q|)^2 \times |q|^{4n}$ possibilities for $w$.
  The sequence $\log(v_n)/n^2$ converges to $0$ as $n$ grows to infinity,
  thus the entropy of $X_q$ is zero.
\end{proof}

\begin{pro}
  \label{pro:interchg}
  Suppose that there exists an interchangeable pair $(p_1,p_2)$ for
  $q$. Suppose further that there are a configuration $w$ in $X_q$,
  strictly positive integers $k,\ell$, and domains $(D_i)_{i \in \nn}$
  such that $\cup_i D_i = \zz^2$, that each $D_i$ is a rectangle of
  size $(k,\ell)$, and that for each $i$ the fragment $w(D_i)$
  contains either an occurrence of $p_1$ or of $p_2$. Then the entropy
  of $X_q$ is at least $(k\ell)^{-1}$.
\end{pro}

\noindent
Note that the condition of this proposition is satisfied, in
particular, when the shape of the interchangeable pair tiles the plane
by translation.

\begin{proof}[of Proposition~\ref{pro:interchg}]
  Let $t$ denote a natural integer and $c$ an arbitrary block of size
  $t \times t$ in $w$. By the assumptions on $w$, there is at least
  $u_t = (\frac{t}{k}-2) \times (\frac{t}{\ell}-2)$ occurrences of
  $\{p_1, p_2\}$ in $c$. Since each occurrence of $p_1$ can be swapped
  to an occurrence of $p_2$, and vice-versa, without leaving $L(X_q)$,
  we have at least $2^{u_t}$ blocks of size $t \times t$ in
  $L(X_q)$. Compute:
  \begin{equation*}
    \lim_{t \to \infty} \frac{\log(2^{u_t})}{t^2} =
    \lim_{t \to \infty} \frac{1}{t^2}(\frac{t^2}{k\ell} - \frac{2t}{k} - \frac{2t}{\ell} + 4) = (k\ell)^{-1}.
  \end{equation*}
  The proposition is proved.
\end{proof}

\section{A sufficient condition for strictly positive entropy}
\label{sec:suf}

Ideally we would like to prove the converse of Theorem~\ref{thm:main},
in order to have a necessary and sufficient condition on the cover to
get a strictly positive entropy coverable subshift. However, it is not
clear whether the condition of this theorem is actually sufficient; it
might be too weak. We define another condition on $q$ which is
stronger than the condition~$(*)$ and which is sufficient for positive
entropy. We prove the following theorem.

\begin{thm}
  \label{thm:main2}
  Let $q$ denote a primitive block with size $(w,h)$. Suppose that $q$
  has two borders $b_1$ and $b_2$ in opposite corners, such that:
  \begin{enumerate}
  \item either $\width(b_1)=\width(b_2)$ and
    $\height(b_1)+\height(b_2) \geq \height(q)$,
  \item or $\height(b_1)=\height(b_2)$ and
    $\width(b_1)+\width(b_2) \geq \width(q)$;
  \end{enumerate}
  then $X_q$ has entropy at least $(9wh)^{-1}$.
\end{thm}

Figure~\ref{fig:cassaigne} illustrates the condition of
Theorem~\ref{thm:main2}.

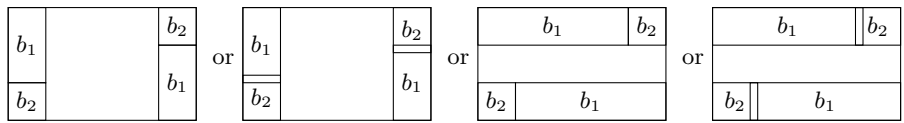
\begin{figure}[htbp]
  \centering
  
  \hfill
  \begin{tikzpicture}[scale=0.5,baseline={(current bounding box.center)}]
    \draw (0,0) rectangle ++(5,3);
    \draw (0,0) rectangle node{$b_2$} ++(1,1);
    \draw (0,1) rectangle node{$b_1$} ++(1,2);
    \draw (4,2) rectangle node{$b_2$} ++(1,1);
    \draw (4,0) rectangle node{$b_1$} ++(1,2);
  \end{tikzpicture}
  \hfill or\hfill
  \begin{tikzpicture}[scale=0.5,baseline={(current bounding box.center)}]
    \draw (0,0) rectangle ++(5,3);
    \draw (0,0) rectangle node{$b_2$} ++(1,1.2);
    \draw (0,1) rectangle node{$b_1$} ++(1,2);
    \draw (4,1.8) rectangle node{$b_2$} ++(1,1.2);
    \draw (4,0) rectangle node{$b_1$} ++(1,2);
  \end{tikzpicture}
  \hfill or\hfill
  \begin{tikzpicture}[scale=0.5,baseline={(current bounding box.center)}]
    \draw (0,0) rectangle ++(5,3);
    \draw (0,0) rectangle node{$b_2$} ++(1,1);
    \draw (1,0) rectangle node{$b_1$} ++(4,1);
    \draw (4,2) rectangle node{$b_2$} ++(1,1);
    \draw (0,2) rectangle node{$b_1$} ++(4,1);
  \end{tikzpicture}
  \hfill or\hfill
  \begin{tikzpicture}[scale=0.5,baseline={(current bounding box.center)}]
    \draw (0,0) rectangle ++(5,3);
    \draw (0,0) rectangle node{$b_2$} ++(1.2,1);
    \draw (1,0) rectangle node{$b_1$} ++(4,1);
    \draw (3.8,2) rectangle node{$b_2$} ++(1.2,1);
    \draw (0,2) rectangle node{$b_1$} ++(4,1);
  \end{tikzpicture}
  \hfill\hspace{0pt}
  \caption{Sufficient condition to have an interchangeable pair. Note
    that $b_1$ and $b_2$ may overlap.}
  \label{fig:cassaigne}
\end{figure}

\begin{proof}[of Theorem~\ref{thm:main2}]
  Without loss of generality, suppose $q$ satisfies Condition~1 of the
  theorem.  Figure~\ref{fig:interpair} shows an interchangeable pair
  for $q$. Indeed, one easily checks that both patterns have the same
  domain and are $q$-coverable. Moreover they are different, otherwise
  we would have
  $(b_1\ominus{}b_2)\obar{}q = q\obar{}(b_2\ominus{}b_1)$, with
  $\ominus$ denoting horizontal concatenation and $\obar$ denoting
  vertical concatenation. This situation implies that $q$ is not
  primitive by~\cite[Theorem~3]{GamardRichommeShallitSmith2017}.  This
  pair tiles the plane (see Figure~\ref{fig:periodic}), so
  Proposition~\ref{pro:interchg} implies that $X_q$ has strictly
  positive entropy. Moreover, the tiling on Figure~\ref{fig:periodic}
  shows that there is at least one occurrence of the pair in each
  rectangle of size $(3w, 3h)$, thence the bound on the entropy.

  \begin{figure}[htbp]
    \begin{minipage}[b]{0.4\linewidth}
      \centering
      \begin{tikzpicture}[scale=0.3] 
        \draw[fill=gray!50!white] (4,3) rectangle node{$q$} ++(5,3);
        \draw (0,1) rectangle node{$q$} ++(5,3);
        \draw (0,4) rectangle node{$q$} ++(5,3);
        \draw (4,0) rectangle node{$q$} ++(5,3);
        \draw (4,6) rectangle node{$q$} ++(5,3);
        \draw (9,2) rectangle node{$q$} ++(5,3);
        \draw (9,5) rectangle node{$q$} ++(5,3);
        \draw (4.5,3.5) node{$b_2$};
        \draw (4.5,2) node{$b_1$};
        \draw (4.5,6.5) node{$b_2$};
        \draw (4.5,5) node{$b_1$};
      \end{tikzpicture}

      \smallskip

      \begin{tikzpicture}[scale=0.3]
        \draw[fill=gray!50!white] (5,3) rectangle node{$q$} ++(5,3);
        \draw (0,1) rectangle node{$q$} ++(5,3);
        \draw (0,4) rectangle node{$q$} ++(5,3);
        \draw (4,0) rectangle node{$q$} ++(5,3);
        \draw (4,6) rectangle node{$q$} ++(5,3);
        \draw (9,2) rectangle node{$q$} ++(5,3);
        \draw (9,5) rectangle node{$q$} ++(5,3);
        \draw (4.5,2) node{$b_1$};
        \draw (4.5,6.5) node{$b_2$};
        \draw (9.5,4) node{$b_1$};
        \draw (9.5,5.5) node{$b_2$};
      \end{tikzpicture}
      \caption{An interchangeable pair for Theorem~\ref{thm:main2}.}
      \label{fig:interpair}
    \end{minipage}
    \hfill
    \begin{minipage}[b]{0.56\linewidth}
      \centering
      $\vdots$
      
      $\dots$
      \begin{tikzpicture}[baseline={(current bounding box.center)},scale=0.25]
        \begin{scope}
          \draw[fill=gray!50!white] (4,3) rectangle ++(5,3);
          \draw (0,1) rectangle ++(5,3);
          \draw (0,4) rectangle ++(5,3);
          \draw (4,0) rectangle ++(5,3);
          \draw (4,6) rectangle ++(5,3);
          \draw (9,2) rectangle ++(5,3);
          \draw (9,5) rectangle ++(5,3);
        \end{scope}
        \begin{scope}[xshift=9cm,yshift=1cm]
          \draw[fill=gray!50!white] (4,3) rectangle ++(5,3);
          \draw (0,1) rectangle ++(5,3);
          \draw (0,4) rectangle ++(5,3);
          \draw (4,0) rectangle ++(5,3);
          \draw (4,6) rectangle ++(5,3);
          \draw (9,2) rectangle ++(5,3);
          \draw (9,5) rectangle ++(5,3);
        \end{scope}
        \begin{scope}[xshift=0cm,yshift=6cm]
          \draw[fill=gray!50!white] (4,3) rectangle ++(5,3);
          \draw (0,1) rectangle ++(5,3);
          \draw (0,4) rectangle ++(5,3);
          \draw (4,0) rectangle ++(5,3);
          \draw (4,6) rectangle ++(5,3);
          \draw (9,2) rectangle ++(5,3);
          \draw (9,5) rectangle ++(5,3);
        \end{scope}
        \begin{scope}[xshift=9cm,yshift=7cm]
          \draw[fill=gray!50!white] (4,3) rectangle ++(5,3);
          \draw (0,1) rectangle ++(5,3);
          \draw (0,4) rectangle ++(5,3);
          \draw (4,0) rectangle ++(5,3);
          \draw (4,6) rectangle ++(5,3);
          \draw (9,2) rectangle ++(5,3);
          \draw (9,5) rectangle ++(5,3);
        \end{scope}
      \end{tikzpicture}
      $\dots$
      
      $\vdots$
      \caption{The interchangeable pair for Theorem~\ref{thm:main2} tiles the plane.}
      \label{fig:periodic}
    \end{minipage}
  \end{figure}
\end{proof}

\section{Conclusion}
\label{sec:conclu}

\smallskip\noindent \textbf{Results.} %
We showed that the set of $q$-coverable $\zz^2$-words is a subshift of
finite type, for all $q$. Then we gave a necessary condition and a
sufficient condition on $q$ for this subshift to have strictly
positive entropy. These conditions were not quite identical, but
close; we also gave a lower bound on the entropy when the sufficient
condition is satisfied. This lower bound used the concept of
interchangeable pair; we showed that any coverable subshift with
strictly positive entropy has interchangeable pairs.

\smallskip\noindent \textbf{Open problems.} %
Our work may be extended in various directions. The first direction is
to generalize the notion of coverability: allow non-rectangular
shapes, allow two covers instead of one (as a disjunction of covers),
or even allow some minimal distance between the covers (``negative
overlaps'', in a sense). Connections with the recurrence function
could be established.

Another direction of further work is to close the gap between our
necessary and our sufficient condition, and to give more precise
bounds on the value of the entropy in function of $q$.

\smallskip\noindent \textbf{Acknowledgements.} %
I would like to thank %
Julien Cassaigne for the helpful discussions that lead to the proof of
Theorem~\ref{thm:main2}; %
Benjamin Hellouin for his simplification of the proof of
Proposition~\ref{pro:qpsft}; %
and %
Mikhail Vyalyi, who read an early draft of this paper and made many
useful suggestions and comments.

\bibliographystyle{splncs04}
\bibliography{ent2d}{}

\begin{thebibliography}{10}
\providecommand{\url}[1]{\texttt{#1}}
\providecommand{\urlprefix}{URL }
\providecommand{\doi}[1]{https://doi.org/#1}

\bibitem{AnselmoGiammarresiMadonia2017}
Anselmo, M., Giammarresi, D., Madonia, M.: Avoiding overlaps in pictures. In:
  Descriptional Complexity of Formal Systems - 19th {IFIP} {WG} 1.02
  International Conference, {DCFS} 2017, July 3--5, Milano, Italy. pp. 16--32
  (2017)

\bibitem{ApostolicoEhrenfeucht1993}
Apostolico, A., Ehrenfeucht, A.: Efficient detection of quasiperiodicities in
  strings. Theor. Comput. Sci.  \textbf{119}(2),  247--265 (1993)

\bibitem{CrochemoreIliopoulosKorda1998}
Crochemore, M., Iliopoulos, C.S., Korda, M.: Two-dimensional prefix string
  matching and covering on square matrices. Algorithmica  \textbf{20}(4),
  353--373 (1998)

\bibitem{GamardRichomme2015}
Gamard, G., Richomme, G.: Coverability in two dimensions. In: Language and
  Automata Theory and Applications - 9th International Conference, {LATA} 2015,
  March 2--6, Nice, France. pp. 402--413 (2015)

\bibitem{GamardRichomme2016}
Gamard, G., Richomme, G.: Determining sets of quasiperiods of infinite words.
  In: 41st International Symposium on Mathematical Foundations of Computer
  Science, {MFCS} 2016, August 22-26, Krak{\'{o}}w, Poland. pp. 40:1--40:13
  (2016)

\bibitem{GamardRichomme2017}
Gamard, G., Richomme, G.: Coverability and multi-scale coverability on infinite
  pictures. Journal of Compututer and System Sciences  (2017), accepted

\bibitem{GamardRichommeShallitSmith2017}
Gamard, G., Richomme, G., Shallit, J., Smith, T.J.: Periodicity in rectangular
  arrays. Information Processing Letters  \textbf{118},  58--63 (2017)

\bibitem{GlenLeveRichomme2008}
Glen, A., Lev{\'{e}}, F., Richomme, G.: Quasiperiodic and {L}yndon episturmian
  words. Theor. Comput. Sci.  \textbf{409}(3),  578--600 (2008)

\bibitem{Gummelt1996}
Gummelt, P.: Penrose tilings as coverings of congruent decagons. Geometriae
  Dedicata  \textbf{62},  1--17 (1996)

\bibitem{Lieb1967}
Lieb, E.: Residual entropy of square ice. Phys. Rev.  \textbf{162}(1),
  162--172 (1967)

\bibitem{Marcus2004}
Marcus, S.: Quasiperiodic infinite words (columns: Formal language theory).
  Bulletin of the {EATCS}  \textbf{82},  170--174 (2004)

\bibitem{MarcusMonteil2006}
Monteil, T., Marcus, S.: Quasiperiodic infinite words: multi-scale case and
  dynamical properties. Arxiv  \textbf{math/0603354} (2006)

\bibitem{MorseHedlund1940}
Morse, M., Hedlund, G.A.: Symbolic dynamics {II}: Sturmian trajectories. Amer.
  J. Math.  \textbf{62}(1),  1--42 (1940)

\bibitem{Penrose1978}
Penrose, R.: Pentaplexity: A class of non-periodic tilings of the plane. Eureka
   \textbf{39} (1978)

\bibitem{Robinson1971}
Robinson, R.M.: Undecidability and nonperiodicity for tilings of the plane.
  Invent. Math.  \textbf{12}(3) (1971)

\end{thebibliography}

\end{document}